\documentclass{article}

\PassOptionsToPackage{numbers, compress, sort}{natbib}
\usepackage[preprint]{neurips_2021}
\usepackage[utf8]{inputenc} % allow utf-8 input
\usepackage[T1]{fontenc}    % use 8-bit T1 fonts
\usepackage{hyperref}       % hyperlinks
\usepackage{url}            % simple URL typesetting
\usepackage{booktabs}       % professional-quality tables
\usepackage{amsmath}
\usepackage{amssymb}
\usepackage{amsfonts}       % blackboard math symbols
\usepackage{amsthm}
\usepackage{thmtools}
\usepackage{thm-restate}
\usepackage{nicefrac}       % compact symbols for 1/2, etc.
\usepackage{microtype}      % microtypography
\usepackage{xcolor}         % colors
\usepackage{graphicx}
\usepackage[ruled, vlined]{algorithm2e}
\usepackage{appendix}
\declaretheorem[name=Theorem, numberwithin=section]{theorem}
\declaretheorem[name=Lemma, sibling=theorem]{lemma}
\declaretheorem[name=Definition, sibling=theorem]{definition}
\declaretheorem[name=Corollary, sibling=theorem]{corollary}

\newcommand{\R}{\mathbb{R}}

\newcommand{\N}{\mathbb{N}}

\newcommand{\dx}{\mathrm{d}}
\newcommand{\calm}{{\mathcal{M}}}
\newcommand{\calx}{{\mathcal{X}}}

\newcommand{\caln}{{\mathcal{N}}}

\newcommand{\cale}{{\mathcal{E}}}

\DeclareMathOperator{\erfc}{erfc}
\DeclareMathOperator{\ban}{Ban}
\DeclareMathOperator{\diag}{diag}
\DeclareMathOperator{\clip}{clip}
\DeclareMathOperator*{\bigtimes}{\times}

\title{Differentially Private Hamiltonian Monte Carlo}

\author{
  Ossi Räisä \quad
  Antti Koskela \quad
  Antti Honkela \\
  Helsinki Institute for Information technology HIIT \\
  Department of Computer Science\\
  University of Helsinki\\
  \texttt{\{ossi.raisa, antti.h.koskela, antti.honkela\}@helsinki.fi} \\
}

\begin{document}

\maketitle

\begin{abstract}
  Markov chain Monte Carlo (MCMC) algorithms have long been the main workhorses
  of Bayesian inference. Among them, Hamiltonian Monte Carlo (HMC) has recently
  become very popular due to its efficiency resulting from effective use of the gradients
  of the target distribution. In privacy-preserving machine learning,
  differential privacy (DP) has become the gold standard in ensuring that
  the privacy of data subjects is not violated. Existing DP MCMC algorithms
  either use random-walk proposals, or do not use the Metropolis--Hastings (MH) acceptance
  test to ensure convergence without decreasing their step size to zero.
  We present a DP variant of HMC using the MH acceptance test that builds on a recently proposed
  DP MCMC algorithm called the penalty algorithm, and adds noise to the gradient evaluations of HMC.
  We prove that the resulting algorithm converges to the correct distribution,
  and is ergodic. We compare DP-HMC with the existing penalty,
  DP-SGLD and DP-SGNHT algorithms, and find that DP-HMC has better or equal performance than
  the penalty algorithm, and performs more consistently than DP-SGLD or DP-SGNHT.
\end{abstract}

\section{Introduction}\label{intro_section}

Differential privacy (DP) \cite{DMN06} has been widely accepted as the
standard approach for developing privacy-preserving algorithms that
guarantee that the output of the algorithm cannot be used to violate
the privacy of the subjects of the input data. Bayesian inference is
one of the widely used approaches for analysis of potentially
sensitive data. In this paper, we present the first DP version of the
modern Bayesian workhorse, Hamiltonian Monte Carlo (HMC)~\cite{DKP87},
with provable convergence to the exact posterior under fixed step
lengths.

HMC is a Markov chain Monte Carlo (MCMC) algorithm that makes use of
gradients of the target density to form a Hamiltonian system that can
be accurately simulated numerically to generate very long jumps with a
high acceptance rate. HMC scales better to higher dimensions than
other MCMC algorithms. New variants \cite{HoG14,HoffmanRS21} that
avoid problem-specific tuning make it an ideal choice for efficient
and accurate black box inference.

Like all MCMC algorithms, HMC requires careful specification of the
algorithm to guarantee convergence to the desired target. This makes
the development of DP MCMC algorithms challenging. The first DP MCMC
algorithms, such as DP stochastic gradient Langevin dynamics (DP-SGLD)
and DP stochastic gradient Nosé-Hoover thermostat (DP-SGNHT),
were based on gradient perturbation for stochastic
gradient MCMC without a Metropolis--Hastings accept/reject
step~\cite{WFS15,LCL19}. These algorithms come with very weak
convergence guarantees requiring decreasing the step size to 0.

The first DP MCMC algorithms implementing an accept/reject step that
enables convergence with fixed step lengths appeared only in 2019
\cite{YildirimE19,HeikkilaJDH19}. Our work builds upon the DP-penalty
algorithm~\cite{YildirimE19} that uses the penalty method \cite{CeD99}
to compensate for the noise added for DP by decreasing the acceptance
rate in a specific way. We adapt the DP-penalty method for HMC, adding
DP gradient evaluations. Our main contribution is the proof that the
resulting algorithm is ergodic and converges to the desired target.

\section{Background}\label{background_section}

In this section, we introduce main background material relevant to our work.
Section~\ref{dp_basics_section} introduces differential privacy and the
privacy accounting method we use. Section~\ref{hmc_section} introduces
MH algorithms and the HMC algorithm. Section~\ref{hmc_convergence_section}
is very technical, and contains the most relevant measure-theoretic background
material for our main theorem, the convergence proof of DP-HMC in
Theorem~\ref{dp_hmc_convergence_theorem}, and a proof that HMC converges to the
correct distribution, which serves as a preliminary to
Theorem~\ref{dp_hmc_convergence_theorem}.

\subsection{Differential Privacy}\label{dp_basics_section}

Differential privacy~\cite{DMN06} (DP) formalises the notion of a privacy-preserving algorithm
by requiring that the distribution of the output only changes slightly
given a change to a single individual's data. Of the many definitions, we use
\emph{approximate DP} (ADP)~\cite{DKM06}, also known as \((\epsilon, \delta)\)-DP:
\begin{definition}
  A mechanism \(\calm\colon \calx \to \R^{d}\) is \((\epsilon, \delta)\)-ADP
  for neighbourhood relation \(\sim\)
  if for all measurable \(S\subset \R^{d}\) and all \(X, X'\in \calx\) with
  \(X\sim X'\),
  \[
    P(\calm(X) \in S) \leq e^{\epsilon}P(\calm(X') \in S) + \delta.
  \]
\end{definition}
We exclusively focus on tabular data and the substitute neighbourhood relation
\(\sim_{S}\) which means that datasets \(X, X'\in \R^{n\times d_{x}}\) are neighbors in
\(\sim_{S}\)-relation, $X \sim_{S} X'$, if they differ in at most one row.
We use \(x\in X\) to denote that \(x\) is a row of \(X\).

DP has two attractive properties:
post-processing immunity means that applying a function to the output
of a DP mechanism does not change the privacy bounds, and composability
means that releasing the output of several DP algorithms together is DP,
although with worse privacy bounds~\cite{DwR14}. 

To make HMC DP, we use the \emph{Gaussian mechanism}~\cite{DKM06}, together
with post-processing immunity and composition.

\begin{definition}\label{gauss_mechanism_definition}
	The Gaussian mechanism with query \(f\colon \calx\to \R^{d}\) and noise
  variance \(\sigma^{2}\) releases a sample from
  \(f(X) + \caln(0, \sigma^{2}I)\) for input \(X\).
\end{definition}
To achieve DP, the output of the query of the Gaussian mechanism must not vary
too much with changing input:
it must have finite \emph{sensitivity}, and less sensitive queries give
smaller privacy bounds.

\begin{definition}\label{sensitivity_definition}
  The \(l_{2}\)-sensitivity of a function \(f\colon \calx \to \R^{d}\) is defined as
	\[
    \Delta_{2}f = \sup_{X\sim X'}||f(X) - f(X')||_{2}.
  \]
\end{definition}
To compute the privacy bounds for compositions of several Gaussian mechanisms, we use
the tight ADP bound of \citet{Sommer2019}:
\begin{restatable}{theorem}{gausscompositiontheorem}\label{gauss_composition_theorem}
  Let \(f_{i}\) be queries with \(\Delta_{2} f_{i} \leq \Delta_{i}\)
  for \(1\leq i \leq k\). Then the composition of \(k\) Gaussian mechanisms
  with queries \(f_{i}\) and noise variances \(\sigma_{i}^{2}\) for
  \(1\leq i \leq k\) is \((\epsilon, \delta(\epsilon))\)-ADP with
  \[
    \delta(\epsilon) = \frac{1}{2}\left(
      \erfc\left(\frac{\epsilon - \mu}{2\sqrt{\mu}}\right)
      -e^{\epsilon}\erfc\left(\frac{\epsilon + \mu}{2\sqrt{\mu}}\right)
    \right), \quad \textrm{where} \quad \mu = \sum_{i=1}^{k} \frac{\Delta_{i}^{2}}{2\sigma_{i}^{2}}.
  \]
\end{restatable}
\begin{proof}
  The claim follows from three theorems of \citet{Sommer2019}:
  first, the \emph{privacy loss distribution} (PLD) of a Gaussian mechanism with
  sensitivity \(\Delta\) and noise variance \(\sigma^{2}\) is
  \(\caln(\mu, 2\mu)\) with \(\mu = \frac{\Delta}{2\sigma^{2}}\).
  Second, the PLD of a composition of several mechanisms is the convolution of the PLDs of the mechanisms
  in the composition, so the PLD of a composition of Gaussian mechanisms
  with PLDs \(\caln(\mu_{i}, 2\mu_{i})\), \(1\leq i\leq k\), is
  \(\caln(\sum_{i=1}^{k}\mu_{i}, 2\sum_{i=1}^{k}\mu_{i})\).
  Finally, a mechanism with a PLD \(\caln(\mu, 2\mu)\) is
  \((\epsilon, \delta(\epsilon))\)-ADP with $\delta(\epsilon)$ given by
  \[
    \delta(\epsilon) = \frac{1}{2}\left(
      \erfc\left(\frac{\epsilon - \mu}{2\sqrt{\mu}}\right)
      -e^{\epsilon}\erfc\left(\frac{\epsilon + \mu}{2\sqrt{\mu}}\right)
    \right).\qedhere
  \]
\end{proof}

In this paper, the query \(f\colon\R^{n\times d_{x}}\to \R^{d}\) is always of the summative form
\(f(X) = \sum_{x\in X}g(x)\) with \(g\colon \R^{d_{x}}\to \R^{d}\), so
\(\Delta_{2}f = \sup_{x, x'\in \R^{d_{x}}}||g(x) - g(x')||_{2}.\) % = \Delta^{*}_{2}g\).
Moreover, we \emph{clip} the output of \(g\) to
have a bounded norm, i.e., instead of the function $g$, we consider the function
$\widetilde{g} = \clip_{b} \circ g$,
where \(\clip_{b}(y) = y\min\{\frac{b}{||y||_{2}}, 1\}\). Then clearly 
$\sup_{x, x'\in \R^{d_{x}}}||\widetilde{g}(x) - \widetilde{g}(x')||_{2} \leq 2b$.
Clipping bounds the sensitivity of $f$ and allows adding less noise to the query for equal $(\epsilon,\delta)$-DP guarantees.

\subsection{Metropolis-Hastings and Hamiltonian Monte Carlo}\label{hmc_section}
\emph{Markov chain Monte Carlo} (MCMC) algorithms sample from a distribution \(\pi\)
of \(\theta\) by forming an ergodic Markov chain that has the invariant
distribution \(\pi\)~\cite{Robert04}. The \emph{Metropolis-Hastings} (MH)~\cite{MRR53, Has70}
algorithm constructs the Markov chain by starting from a given point \(\theta_{0}\),
generating \(\theta_{i+1}\) given \(\theta_{i} = \theta\) by sampling a proposal
\(\theta'\) from a proposal distribution \(q(\theta'\mid \theta)\),
and accepting the \(\theta'\) with probability
\[
  \alpha(\theta, \theta')
  = \min\left\{1, \frac{\pi(\theta')}{\pi(\theta)}
    \frac{q(\theta\mid \theta')}{q(\theta'\mid \theta)} \right\}.
\]
If \(\theta'\) is accepted, \(\theta_{i+1} = \theta'\), otherwise \(\theta_{i+1} = \theta\).

The invariant distribution of an MH
algorithm is always \(\pi\), but the ergodicity of the resulting Markov chain
depends on the proposal. A convenient sufficient condition for ergodicity
is \emph{strong irreducibility}: if the proposal can propose any state from any other
state with positive probability, the chain is said to be strongly
ergodic, and thus irreducible~\cite{Robert04}.

MH is commonly used to sample from the posterior \(p(\theta\mid X)\) of a Bayesian
inference problem given by Bayes' theorem:
% \begin{equation}
\[
  p(\theta\mid X) = \frac{p(X\mid \theta)p(\theta)}{p(X)},
\]
% \end{equation}
where \(\theta\in \R^{d}\) denotes the parameters of interest and \(X\) denotes the observed data.
For the MH algorithm, we set \(\pi(\theta) = p(\theta\mid X)\), and the
denominator in Bayes' theorem cancels out in
\(\alpha(\theta, \theta')\), so it is sufficient to consider
\(\pi(\theta)\propto p(X\mid \theta)p(\theta)\) for the MH algorithm.
Usually \(X\in \R^{n\times d_{x}}\) with each row of \(X\) representing a data point,
and the likelihood is \(p(X\mid \theta) = \prod_{x\in X}p(x\mid \theta)\), where
\(x\in X\) means that \(x\) is a row of \(X\).

The \emph{Hamiltonian Monte Carlo} (HMC)~\cite{DKP87, neal2012mcmc}
algorithm is an MH algorithm that generates proposals deterministically
through simulating Hamiltonian dynamics. The dynamics are given by the \emph{Hamiltonian}
\(H(\theta, p) = U(\theta) + \frac{1}{2}p^{T}M^{-1}p\), where \(p\in \R^{d}\) is an
auxiliary momentum variable, \(M\in \R^{d\times d}\) is a positive-definite mass matrix,
and \(U(\theta) = -\ln \pi(\theta)\). The simulation is then given by
Hamilton's equations
\(
  \frac{\dx \theta}{\dx t} = M^{-1}p,
  \frac{\dx p}{\dx t} = -\nabla U(\theta)
\).
Solving them exactly is rarely possible, so in practice
the simulation is carried out using \emph{leapfrog simulation}, given
for a step-size \(\eta > 0\) by
\[
  l = l_{p_{\nicefrac{\eta}{2}}} \circ l_{\theta}\circ l_{p_{\eta}}
  \dotsb \circ l_{p_{\eta}} \circ l_{\theta}\circ l_{p_{\nicefrac{\eta}{2}}},
\]
where
% \begin{align*}
\[
  l_{p_{s}}(\theta, p) = (\theta, p - s\nabla U(\theta)), \quad
  l_{\theta}(\theta, p) = (\theta + \eta M^{-1}p, p).
\]
% \end{align*}

The definition of \(U\) means that \(\pi\) is required to be continuous,
supported on \(\R^{d}\), and have a differentiable
log-density~\cite{neal2012mcmc}. With the auxiliary variable \(p\), HMC targets
the distribution
\[
  \pi^{*}(\theta, p) \propto \exp(-H(\theta, p))
  = \exp(-U(\theta))\exp\left(-\frac{1}{2}p^{T}M^{-1}p\right),
\]
so the marginal distributions of \(\theta\) and \(p\) are independent,
the marginal of \(\theta\) is \(\pi\), and the marginal
of \(p\) is a \(d\)-dimensional Gaussian with mean \(0\) and covariance \(M\).

Proposing a new sample is done in two steps, both of which having a separate
MH acceptance test. First, \(p\) is sampled from its marginal distribution,
which is always accepted. Second, the leapfrog simulation is run and the
final value of \(p\) is negated, which gives a proposal for \((\theta, p)\).
The acceptance probability for the second step is
\[
  \alpha(\theta, p, \theta', p')
  = \min\{1, \exp(H(\theta, p) - H(\theta', p'))\}.
\]
In Section~\ref{hmc_convergence_section}, we will show that this acceptance
probability for the second step makes \(\pi^{*}\) the invariant distribution.
The proof requires some machinery from measure theory,
which is briefly introduced in Section~\ref{hmc_convergence_section},
and serves as a preliminary to our main result, the DP-HMC convergence proof, in
Section~\ref{dp_hmc_section}.

\subsection{Convergence of HMC}\label{hmc_convergence_section}

The proofs of convergence for HMC in Theorem~\ref{hmc_convergence_theorem}
and for DP-HMC in Theorem~\ref{dp_hmc_convergence_theorem} require some
theory of \emph{Markov kernels} and their \emph{reversibility}~\cite{Cin11},
presented in this section. We defer all proofs to either
Appendix~\ref{measure_theory_section}, or the textbook of \citet{Cin11},
with the exception of the proof of Theorem~\ref{hmc_convergence_theorem}, which is fairly
short and serves as a preliminary to the proof of our main result in
Theorem~\ref{dp_hmc_convergence_theorem}.

Recall that a \emph{measurable space} \((E, \cale)\) is a pair of a set \(E\) and
a \(\sigma\)-algebra \(\cale\), and an \emph{involution} is a function \(f\)
with \(f^{-1} = f\).

\begin{restatable}{definition}{markovkerneldefinition}\label{markov_kernel_definition}
  Let \((E, \cale)\) be a measurable space and let
  \(q\colon E\times \cale \to [0, 1]\). \(q\) is called a Markov kernel
  on \((E, \cale)\) if
  \begin{enumerate}
    \item For all \(B\in \cale\), the function \(q(\cdot, B)\) is measurable.
    \item For all \(a\in E\), the function \(q(a, \cdot)\) is a probability
          measure.
  \end{enumerate}
\end{restatable}

Markov kernels are the measure-theoreric formulation of random functions.
The involutiveness of deterministic functions generalises to reversibility
of Markov kernels, as seen in Lemma~\ref{dirac_reversible_lemma}.

\begin{restatable}{definition}{markovkernelreversible}\label{markov_kernel_reversible_definition}
	Let \(q\) be a Markov kernel and let \(\mu\) be a \(\sigma\)-finite measure, both on
  \((E, \cale)\). If
  \[
    \int_{A}\mu(\dx a)\int_{B}q(a, \dx b) = \int_{B}\mu(\dx b)\int_{A}q(b, \dx a)
  \]
  for all \(A, B\in \cale\), \(q\) is said to be reversible with respect to
  \(\mu\).
\end{restatable}
Definition~\ref{markov_kernel_reversible_definition} can be seen as an equality
of two measures using a lemma from measure theory:
\begin{lemma}\label{markov_kernel_product_unique_lemma}
	Let \((E, \cale)\) be a measurable space and let \(q\) be a Markov kernel
  and \(\mu\) be a \(\sigma\)-finite measure, both on \((E, \cale)\). Then there exists a unique
  \(\sigma\)-finite measure \(\nu\) on \((E, \cale)^{2}\) such that
  \[
    \nu(A\times B) = \int_{A}\mu(\dx a)\int_{B}q(a, \dx b)
  \]
  for all \(A, B\in \cale\).
\end{lemma}
\begin{proof}
	See \citet[Theorem I.6.11]{Cin11}.
\end{proof}
Using the uniqueness in Lemma~\ref{markov_kernel_product_unique_lemma},
the equality in Definition~\ref{markov_kernel_reversible_definition} can be
stated as an equality of measures: for a measurable space \((E, \cale)\), setting
\[
  \nu_{1}(A\times B) = \int_{A}\mu(\dx a)\int_{B}q(a, \dx b),
\]
\[
  \nu_{2}(A\times B) = \int_{B}\mu(\dx b)\int_{A}q(b, \dx a)
\]
for all \(A, B\in \cale\) defines unique measures \(\nu_{1}\) and \(\nu_{2}\)
on \((E, \cale)^{2}\). Definition~\ref{markov_kernel_reversible_definition}
is then equivalent to \(\nu_{1} = \nu_{2}\).

As Markov kernels represent randomised functions, they can be composed with each
other, with the composition being another Markov kernel:
\begin{lemma}\label{markov_kernel_composition_lemma}
	The composition of Markov kernels \(q_{1}\) and \(q_{2}\) on a measurable
  space \((E, \cale)\) is a Markov kernel given by
  \[
    (q_{2}\circ q_{1})(a, C) = \int_{E}q_{1}(a, \dx b)q_{2}(b, C)
    = \int_{E}q_{1}(a, \dx b)\int_{C}q_{2}(b, \dx c)
  \]
  for any \(C\in \cale\).
\end{lemma}
\begin{proof}
  See~\citet[Remark I.6.4]{Cin11}.
\end{proof}

A composition of reversible Markov kernels is not itself reversible, but
it does have a closely related property that implies reversibility if the composition
is symmetric:
\begin{restatable}{lemma}{markovkernelreversiblecomposition}\label{markov_kernel_reversible_composition_lemma}
	Let \(q_{1},\dotsc, q_{k}\) be Markov kernels on \((E, \cale)\)
  reversible with respect to a \(\sigma\)-finite measure \(\mu\) on
  \((E, \cale)\). Then
  \[
    \int_{A}\mu(\dx a)\int_{C}(q_{k}\circ \dotsb \circ q_{1})(a, \dx c)
    = \int_{C}\mu(\dx c)\int_{A}(q_{1}\circ \dotsb \circ q_{k})(c, \dx a)
  \]
  for all \(A, C\in \cale\).
\end{restatable}

The proposal of an MH algorithm is a Markov kernel. If it is reversible
with respect to the Lebesgue measure and the target distribution is continuous,
the Hastings correction term \(\frac{q(\theta\mid \theta')}{q(\theta'\mid \theta)} = 1\):

\begin{restatable}{lemma}{mhreversibleproposal}\label{mh_reversible_proposal_lemma}
	If the proposal Markov kernel \(q\) of an MH algorithm is reversible
  with respect to the Lebesgue measure and the target distribution
  \(\pi\) is continuous, using
  \[
    \alpha(\theta, \theta') = \min\left\{1, \frac{\pi(\theta')}{\pi(\theta)}\right\}
  \]
  as the acceptance probability leaves the target \(\pi\) invariant.
\end{restatable}

For a deterministic proposal \(f\), like the HMC leapfrog, the Markov kernel of the
proposal is a \emph{Dirac measure} \(\delta_{f(\theta)}(B) = 1_{B}(f(\theta))\) for
\(\theta\in \R^{d}\) and measurable \(B\subset \R^{d}\).
It turns out that \(\delta_{f(\theta)}\) is a reversible Markov kernel for a
suitable \(f\):
\begin{restatable}{lemma}{diracreversible}\label{dirac_reversible_lemma}
	Let \(f\colon \R^{d}\to \R\) be an involution that preserves Lebesgue measure.
  Then the Dirac measure \(\delta_{f(a)}\), seen as a Markov kernel
  \(q(a, B) = \delta_{f(a)}(B)\), is reversible with respect to the Lebesgue measure.
\end{restatable}

The invariance of the target distribution for HMC follows from
Lemmas~\ref{mh_reversible_proposal_lemma} and \ref{dirac_reversible_lemma}.
\begin{restatable}{theorem}{hmcconvergencetheorem}\label{hmc_convergence_theorem}
	For a continuous distribution \(\pi\) that is supported on \(\R^d\) and
  has a differentiable log-density, if
  \[
    \alpha(\theta, p, \theta', p')
    = \min\{1, \exp(H(\theta, p) - H(\theta', p'))\}
  \]
  is used as the acceptance probability for HMC, the invariant distribution is
  \(\pi^{*}(\theta, p) \propto \exp(-H(\theta, p))\).
\end{restatable}
\begin{proof}
  The proposal for the second step is given by \(l_{-}\circ l\), where
  \(l_{-}(\theta, p) = (\theta, -p)\).
  As \((l_{-}\circ l)^{-1} = l_{-}\circ l\)~\cite{neal2012mcmc} and
  each of \(l_{-}\), \(l_{p_{s}}\)
  and \(l_{\theta}\) preserve Lebesgue measure, the HMC proposal Markov kernel
  \(\delta_{(l_{-}\circ l)(\theta, p)}\) is reversible with respect to the
  Lebesgue measure by Lemma~\ref{dirac_reversible_lemma}. Then, by
  Lemma~\ref{mh_reversible_proposal_lemma}, \(\pi^{*}\) is the invariant distribution of
  HMC.
\end{proof}

Showing that HMC is ergodic is much harder due to the
deterministic proposal, but it can be shown that HMC is ergodic with mild
assumptions on \(U\)~\cite{DMS20}.

\section{DP-HMC}\label{dp_hmc_section}

The DP-penalty algorithm of \citet{YildirimE19} makes the MH acceptance test
private by adding Gaussian noise to the log-likelihood ratio
\(\lambda(\theta, \theta') = \ln\frac{p(X\mid \theta')p(\theta')}{p(X\mid \theta)p(\theta)}\).
They correct the MH acceptance probability with the penalty algorithm~\cite{CeD99},
that changes the acceptance probability to
\[
  \alpha(\theta, \theta') = \min\left\{1,
    \exp\left(\lambda(\theta, \theta') + \xi
      + \ln \frac{q(\theta\mid \theta')}{q(\theta'\mid \theta)}
      - \frac{1}{2}\sigma_{l}^{2}(\theta, \theta')\right)\right\},
\]
where \(\xi \sim \caln(0, \sigma_{l}^{2}(\theta, \theta'))\) is the Gaussian noise
added to the log likelihood ratio. For the DP-penalty algorithm,
\(\sigma_{l}(\theta, \theta') = 2\tau b_{l}||\theta - \theta'||_{2}\),
\(b_{l}||\theta - \theta'||_{2}\) is
the log-likelihood ratio clip bound and \(\tau > 0\) controls the amount of noise.

The privacy bounds for the algorithm are given by
Theorem~\ref{gauss_composition_theorem} with \(\mu_{i} = \frac{1}{2\tau^{2}}\).
The convergence of the penalty algorithm requires that the log-likelihood ratios
are not actually clipped, which can only be ensured on some models, like
Bayesian logistic regression~\cite{YildirimE19}. However, in our experiments
shown in Section~\ref{clipping_experiment_section},
small amounts of clipping did not affect the resulting posterior.

\citet{YildirimE19} only used the Gaussian
distribution as the proposal, but the DP-penalty algorithm does not require
any particular proposal distribution \(q\).
However, if \(q\) depends on
the private data \(X\), both sampling \(q\) and computing
\(\ln\frac{q(\theta\mid \theta')}{q(\theta'\mid \theta)}\) may have a privacy cost
that must be taken into account.

In non-DP HMC, the proposal is the deterministic leapfrog simulation, which
can be made DP by simply clipping the gradients of the
log-likelihood and adding Gaussian noise.

In Theorem~\ref{dp_hmc_convergence_theorem},
we show that applying the penalty correction to the HMC acceptance probability
from Theorem~\ref{hmc_convergence_theorem} results in the correct invariant distribution when
using noisy and clipped gradients in the leapfrog simulation.
We also prove the ergodicity of DP-HMC, which
turns out to be much easier because of the noisy leapfrog, in
Theorem~\ref{dp_hmc_ergodic_theorem}.

In the noisy and clipped leapfrog simulation, the momentum update changes to
\[
  l_{p_{s}}(\theta, p) = (\theta, p - s (g(\theta) + \xi)),
\]
where
\[
  g(\theta) = \sum_{x\in X} \clip_{b}(\nabla \ln p(x \mid \theta))
  + \nabla \ln p(\theta)
\]
and \(\xi \sim \caln(0, \sigma_{g}^{2})\). The noisy and clipped leapfrog is then
\[
  l = l_{p_{\nicefrac{\eta}{2}}} \circ l_{\theta}\circ l_{p_{\eta}}
  \dotsb \circ l_{p_{\eta}} \circ l_{\theta}\circ l_{p_{\nicefrac{\eta}{2}}}.
  \label{dp_hmc_leapfrog_equation}
\]
As \(l_{-}\) is an involution, \(l_{-}\circ l\) can be decomposed as
\[
  l_{-}\circ l =
  (l_{-}\circ p_{\nicefrac{\eta}{2}}) \circ (l_{\theta}\circ l_{-})\circ (l_{-}\circ l_{p_{\eta}})
  \dotsb \circ (l_{-}\circ l_{p_{\eta}})
  \circ (l_{\theta}\circ l_{-})\circ (l_{-}\circ l_{p_{\nicefrac{\eta}{2}}}).
\]
Denoting \(l_{p_{s}}^{-} = l_{-}\circ l_{p_{s}}\) and
\(l_{\theta}^{-} = l_{\theta}\circ l_{-}\), the decomposition can be written as
\[
  l_{-}\circ l = l_{p_{\nicefrac{\eta}{2}}}^{-}\circ l_{\theta}^{-}\circ l_{p_{\eta}}^{-}\circ
  \dotsb \circ l_{p_{\eta}}^{-}\circ l_{\theta}^{-}\circ l_{p_{\nicefrac{\eta}{2}}}^{-}.
\]
This form makes showing that DP-HMC has the correct invariant distribution
convenient.

\begin{restatable}{lemma}{dphmcleapfrogstepreversibility}\label{dp_hmc_leapfrog_step_reversibility_lemma}
	The Markov kernels \(l_{p_{\nicefrac{\eta}{2}}}^{-}\), \(l_{p_{\eta}}^{-}\) and \(l_{\theta}^{-}\) are reversible
  with respect to the Lebesgue measure.
\end{restatable}
\begin{proof}
	The proof is fairly technical, requiring some machinery from measure theory,
  and is deferred to Appendix~\ref{measure_theory_section}.
\end{proof}

\begin{corollary}\label{dp_hmc_leapfrog_reversibility_lemma}
  The Markov kernel \(l_{-}\circ l\) is reversible with respect to the Lebesgue
  measure.
\end{corollary}
\begin{proof}
  By Lemma~\ref{dp_hmc_leapfrog_step_reversibility_lemma}, the decomposition
  \(
    l_{-}\circ l = l_{p_{\nicefrac{\eta}{2}}}^{-}\circ l_{\theta}^{-}\circ l_{p_{\eta}}^{-}\circ \dotsb \circ l_{p_{\eta}}^{-}\circ l_{\theta}^{-}\circ l_{p_{\nicefrac{\eta}{2}}}^{-}
  \)
  fulfills the assumptions of Lemma~\ref{markov_kernel_reversible_composition_lemma}.
  As the decomposition is symmetric, Lemma~\ref{markov_kernel_reversible_composition_lemma}
  then implies that \(l_{-}\circ l\) is reversible with respect to the Lebesgue
  measure.
\end{proof}

\begin{theorem}\label{dp_hmc_convergence_theorem}
  For a continuous distribution \(\pi\) that is supported on \(\R^{d}\) and has
  a differentiable log-likelihood, if
  \[
    \alpha_{DP}(\theta, p, \theta', p') =
    \min\left\{1, \exp\left(H(\theta, p) - H(\theta', p') + \xi
        - \frac{1}{2}\sigma_{l}^{2}(\theta, \theta')\right)\right\},
  \]
  where \(\xi \sim \caln(0, \sigma_{l}^{2}(\theta, \theta'))\), is used as the
  acceptance probability of DP-HMC and log-likelihood ratios are not clipped,
  the invariant distribution is \(\pi^{*}(\theta, p)\propto \exp(-H(\theta, p))\).
\end{theorem}
\begin{proof}
  By Corollary~\ref{dp_hmc_leapfrog_reversibility_lemma} and
  Lemma~\ref{mh_reversible_proposal_lemma}, using
  \(l_{-}\circ l\) as the proposal of an MH algorithm with
  \[
    \alpha(\theta, p, \theta', p')
    = \min\left\{1, \exp\left(H(\theta, p) - H(\theta', p')\right)\right\}
  \]
  as the acceptance probability makes \(\pi^{*}\) the invariant distribution of
  the algorithm. Applying the DP-penalty algorithm to \(\alpha\) results in the
  acceptance probability
  \[
    \alpha_{DP}(\theta, p, \theta', p') =
    \min\left\{1, \exp\left(H(\theta, p) - H(\theta', p') + \xi
        - \frac{1}{2}\sigma_{l}^{2}(\theta, \theta')\right)\right\},
  \]
  where \(\xi \sim \caln(0, \sigma_{l}^{2}(\theta, \theta'))\), leaving \(\pi^{*}\) as the invariant
  distribution.
\end{proof}

Like the DP-penalty algorithm, Theorem~\ref{dp_hmc_convergence_theorem}
assumes that the log-likelihood ratio is not clipped. This means that convergence
is not guaranteed in the presence of clipping, but in practice, we found that
clipping a small percentage of the log-likelihood ratios does not affect the
resulting posterior, as presented in Section~\ref{clipping_experiment_section}.
Clipping gradients does not affect convergence, but it
likely lowers the acceptance rate, thus reducing the utility of any sample.

\begin{theorem}\label{dp_hmc_ergodic_theorem}
	DP-HMC is strongly irreducible, and thus ergodic.
\end{theorem}
\begin{proof}
  Consider the last four updates of the leapfrog proposal for \(L > 1\),
  \(l_{-}\circ l_{p_{\nicefrac{\eta}{2}}}\circ l_{\theta}\circ l_{p_{\eta}}\). If \(L = 1\),
  the first of them will be \(l_{p_{\nicefrac{\eta}{2}}}\) instead, which does not affect the
  proof. Denote
  \begin{align*}
    (\theta_{1}, p_{1}) = l_{p_{\eta}}(\theta_{0}, p_{0}),\quad
    (\theta_{2}, p_{2}) = l_{\theta}(\theta_{1}, p_{1}),\quad
    (\theta_{3}, p_{3}) = l_{p_{\nicefrac{\eta}{2}}}(\theta_{2}, p_{2}),\quad
    (\theta_{4}, p_{4}) = l_{-}(\theta_{3}, p_{3}).
  \end{align*}
  Now \(\theta_{2} = \theta_{1} + \eta M^{-1}p_{1}\).
  As \(p_{1} \sim \caln(p_{0} - \eta g(\theta_{0}), \eta^{2}\sigma_{g}^{2})\),
  and as \(M\) is non-singular, \(\eta M^{-1}p_{1}\) has a Gaussian distribution
  with support \(\R^{d}\). As \(\theta_{1} = \theta_{0}\) and
  \(\theta_{4} = \theta_{3} = \theta_{2}\), it is possible to obtain any
  value for \(\theta_{4}\) no matter the starting point \((\theta_{0}, p_{0})\).

  Additionally,
  \(p_{4} = -p_{3} \sim \caln(p_{2} - \frac{\eta}{2}g(\theta_{2}), \frac{\eta^{2}}{4}\sigma_{g}^{2})\),
  so it is possible to obtain any \(p_{4}\) given any \((\theta_{2}, p_{2})\).
  Together, these observations mean that it is possible to obtain any
  value of \((\theta_{4}, p_{4})\) given any starting point \((\theta_{0}, p_{0})\).
  This implies that DP-HMC is strongly irreducible, and thus ergodic~\cite{Robert04}.
\end{proof}

For non-DP HMC, it is standard practice to perturb \(\eta\) between iterations
to help the algorithm escape areas where the leapfrog simulation circles back
near the starting point that may occur if both \(\eta\) and \(L\) are kept
constant~\cite{neal2012mcmc}. As \(\eta\) will be constant during each
leapfrog simulation, this does not affect the invariant distribution of the
algorithm. For DP-HMC, we use a randomised Halton sequence~\cite{Owen17tech} to
perturb \(\eta\) after \citet{HoffmanRS21}, although this may not be
as necessary in DP-HMC as the leapfrog simulation is already noisy.

\begin{algorithm}
  \KwIn{
    likelihood \(p(x\mid \theta)\), prior \(p(\theta)\),
    data \(X\), noise parameters \(\tau_{l}\) and
    \(\tau_{g}\), clip bounds \(b_{l}\) and \(b_{g}\), number of iterations
    \(k\), step size sequence \(\eta_{i}\) for \(1\leq i \leq k\),
    number of leapfrog steps \(L\), positive-definite mass matrix \(M\),
    initial value \(\theta_{0}\).
  }
  \(c_{l}(\theta, \theta') = b_{l}||\theta - \theta'||_{2}\)\;
  \(\sigma_{l}^{2}(\theta, \theta') = 4\tau_{l}^{2}c_{l}^{2}(\theta, \theta')\)\;
  \(\sigma_{g}^{2} = 4b_{g}^{2}\tau_{g}^{2}\)\;
  \For{\(1 \leq i \leq k\)}{
    \(\theta = \theta_{i-1},\quad\) \(\theta' = \theta\)\;
    Sample \(p \sim \caln_{d}(0, M)\) and set \(p' = p\)\;

    % Define \(l\) by Equation~\eqref{dp_hmc_leapfrog_equation} with \(l_{\theta}\)
    % appearing \(L\) times\;
    % \((\theta', p') = (l_{-}\circ l)(\theta, p)\)\;
    \((\theta', p') = l_{p_{\nicefrac{\eta_{i}}{2}}}(\theta', p')\)\;
    \For{\(1 \leq j \leq L - 1\)}{
      \((\theta', p') = l_{\theta}(\theta', p')\)\;
      \((\theta', p') = l_{p_{\eta_{i}}}(\theta', p')\)\;
    }
    \((\theta', p') = l_{\theta}(\theta', p')\)\;
    \((\theta', p') = l_{p_{\nicefrac{\eta_{i}}{2}}}(\theta', p')\)\;

    \(r_{x} = \ln\frac{p(x\mid \theta')}{p(x\mid \theta)}\)\;
    \(R = \sum_{x\in X}\clip_{c_{l}(\theta, \theta')}(r_{x})\)\;
    Sample \(\xi\sim \caln(0, \sigma_{l}^{2}(\theta, \theta'))\)\;
    \(\Delta p = \frac{1}{2}p^{T}M^{-1}p - \frac{1}{2}p'^{T}M^{-1}p'\)\;
    \(\Delta H = R + \Delta p + \ln\frac{p(\theta')}{p(\theta)} + \xi\)\;
    Sample \(u\sim \mathrm{Unif}(0, 1)\)\;
    \eIf{\(\ln u < \Delta H - \frac{1}{2}\sigma_{l}^{2}(\theta, \theta')\)}{
      \(\theta_{i} = \theta'\)\;
    } {
      \(\theta_{i} = \theta\)\;
    }
  }
  \Return \((\theta_{1},\dotsc, \theta_{k})\)\;
  \caption{DP-HMC}
  \label{dp_hmc_algo}
\end{algorithm}

Algorithm~\ref{dp_hmc_algo} presents DP-HMC.
In Algorithm~\ref{dp_hmc_algo}, the gradient \(\nabla U\) is evaluated
\(L + 1\) times per iteration of the outer for-loop, for a total of \(k(L+1)\)
times, and the log-likelihood ratio is evaluated \(k\) times in total.
The privacy cost can then be computed from Theorem~\ref{gauss_composition_theorem}:
\begin{theorem}\label{dp_hmc_privacy_theorem}
	DP-HMC (Algorithm~\ref{dp_hmc_algo}) is \((\epsilon, \delta(\epsilon))\)-ADP for substitute
  neighbourhood for
  % \[
%     \mu = \frac{k}{2\tau_{l}^{2}} + \frac{k(L + 1)}{2\tau_{g}^{2}}
%   \]
  \[
    \delta(\epsilon) = \frac{1}{2}\left(
      \erfc\left(\frac{\epsilon - \mu}{2\sqrt{\mu}}\right)
      -e^{\epsilon}\erfc\left(\frac{\epsilon + \mu}{2\sqrt{\mu}}\right)
    \right), \quad \textrm{where} \quad \mu = \frac{k}{2\tau_{l}^{2}} + \frac{k(L + 1)}{2\tau_{g}^{2}}.
  \]
\end{theorem}
\begin{proof}
  The sensitivity of the log-likelihood ratio is \(2b_{l}||\theta - \theta'||_{2}\)
  and the sensitivity of the gradient is \(2b_{g}\). Thus, 
  %for Theorem~\ref{gauss_composition_theorem}, 
  adding noise with
  variance \(\sigma_{l}^{2}(\theta, \theta')\) to the log-likelihood ratio
  gives a sensitivity-variance ratio \(\mu_{l} = \frac{1}{2\tau_{l}^{2}}\). Adding noise with variance
  \(\sigma_{g}^{2}\) to the gradients has sensitivity-variance ratio \(\mu_{g} = \frac{1}{2\tau_{l}^{2}}\).
  As the log-likelihood ratio is evaluated \(k\) times and the gradients are evaluated
  \(k(L + 1)\) times, the total \(\mu\) in Theorem~\ref{gauss_composition_theorem}
  is
  \[
    \mu = k\mu_{l} + k(L+1)\mu_{g} = \frac{k}{2\tau_{l}^{2}} + \frac{k(L + 1)}{2\tau_{g}^{2}}.\qedhere
  \]
\end{proof}

It is possible to shave off one gradient evaluation per outer for-loop iteration of
Algorithm~\ref{dp_hmc_algo}, except the first one, by observing that the first gradient
evaluation computed during an outer for-loop iteration is the same gradient as either the
first for rejected proposals, or the last for accepted proposals, gradient evaluation of the
previous iteration. However, this causes the current iteration to depend on the noise
value generated for that gradient evaluation during the previous iteration, so it is
not clear whether the resulting chain is Markov. As the potential privacy cost saving from this
optimisation is small, we did not investigate this further.

\section{Experiments}\label{experiments_section}

We ran comparisons on two synthetic posterior distributions, presented in
Section~\ref{comparison_section}: a 10-dimensional correlated
Gaussian model and a banana distribution model that results in a non-convex
banana shaped posterior. We also experimented with the effect of clipping
log-likelihood ratios, presented in Section~\ref{clipping_experiment_section}.
The code for the experiment is publicly available.\footnote{
  \url{https://github.com/DPBayes/DP-HMC-experiments}
}

\paragraph{Gaussian Model}
The Gaussian is a 10-dimensional model where the prior and likelihood
for parameters \(\theta\in \R^{d}\) and \(X\in \R^{n\times d}\) are given by
\(\theta \sim \caln_{d}(\mu_{0}, \sigma_{0}^{2}I)\) and \(x \sim \caln_{d}(\theta, \Sigma)\),
where \(\mu_{0}\) and \(\sigma_{0}\) are the prior hyperparameters, and \(\Sigma\)
is the known variance.
As the prior is a Gaussian distribution, the posterior is also a Gaussian
with known analytical form~\cite{BDA}. We used \(d = 10\), \(n = 100000\),
\(\mu_{0} = 0\), \(\sigma_{0} = 100\). \(\Sigma\) was chosen after \citet{HoffmanRS21} by
sampling its eigenvalues from a gamma distribution with shape parameter \(0.5\)
and scale parameter 1, and sampling the eigenvectors by orthonormalising the columns of a
random matrix with each entry sampled from the uniform distribution on \([0, 1]\).

\paragraph{Banana Model}
The banana distribution~\cite{TPK14} is a probability distribution in the shape
of a banana that is a challenging target for MCMC algorithms due to its
non-convex and thin shape. The distribution is a transformation of the
2-dimensional Gaussian distribution using the function
\(g(x_{1}, x_{2}) = (x_{1}, x_{2} - a x_{1}^{2})\). If \(x\sim \caln_{2}(\mu, \Sigma)\),
\(g(x)\) has the banana distribution denoted by \(\ban(\mu, \Sigma, a)\).
To test DP algorithms, we need a Bayesian inference problem where the posterior
is a banana distribution. This is given by a transformation of the Gaussian model:
\begin{align*}
	\theta = (\theta_{1}, \theta_{2}) \sim \ban(0, \sigma_{0}^{2}I, a), \quad
  x_{1} \sim \caln(\theta_{1}, \sigma_{1}^{2}), \quad
  x_{2} \sim \caln(\theta_{2} + a\theta_{1}^{2}, \sigma_{2}^{2}).
\end{align*}
The posterior of this model for data \(X\in \R^{n\times 2}\) is
\(\ban(\mu, \Sigma, a)\), where,
denoting \(\tau_{i} = \frac{1}{\sigma_{i}^{2}}\) and
\(\bar{x}_{i} = \frac{1}{n}\sum_{j=1}^{n}X_{ji}\),
\[
  \mu = \left(
    \frac{n\tau_{1}\bar{x}_{1}}{n\tau_{1} + \tau_{0}},
    \frac{n\tau_{2}\bar{x}_{2}}{n\tau_{2} + \tau_{0}}
  \right),
  \quad
  \Sigma = \diag\left(
    \frac{1}{n\tau_{1} + \tau_{0}},
    \frac{1}{n\tau_{2} + \tau_{0}}
  \right).
\]
We used the hyperparameter values \(\sigma_{0} = 1000\), \(\sigma_{1}^{2} = 2000\),
\(\sigma_{2}^{2} = 2500\) and \(a = 20\), \(n = 100000\) and true parameter
values \(\theta_{1} = 0\), \(\theta_{2} = 3\).

\paragraph{Evaluation}
Our main evaluation metric is
\emph{maximum mean discrepancy} (MMD)~\cite{GrettonBRSS12}, which measures the distances
between distributions, and can be estimated from a sample of both distributions.
We used a Gaussian kernel, and chose the kernel width by choosing a 500 point
subsample from both samples with replacement, and used the median between the
distances of both subsamples. Additionally, we plot the distance of the
mean of the chain and the true posterior sample mean as a more interpretable
evaluation metric.

\subsection{Comparison of DP-MCMC Algorithms}\label{comparison_section}

\paragraph{Detailed implementation}
We compare DP-HMC with DP-penalty~\cite{YildirimE19}, DP-SGLD~\cite{WFS15, LCL19}
and DP-SGNHT~\cite{WFS15, DFB14}.
For both models, we ran all algorithms with 4 chains, started from separate starting points.
The starting points were chosen by sampling a point from a Gaussian distribution centered on
the true parameter values, with standard deviation equal to the mean of the
componentwise standard deviations of the reference posterior sample. Each run
was repeated 10 times with different with different starting points, but
each algorithm and value of \(\epsilon\) had the same set of starting points.
Algorithm parameters were tuned by manually examining diagnostics from preliminary runs.

Our method of picking starting points close to the area of high probability favors
DP-penalty, as the gradient-based methods can use the gradient to quickly find
the area of high probability, even when starting far from it. On the other hand,
it simulates the effect of finding the rough location of the posterior
through another method, such as a MAP estimate or variational inference,
with a very small privacy budget, which \citet{HeikkilaJDH19} used in their experiments.

We combined the samples from all 4 chains, discarded the first half as
warmup samples, and
compared them against 1000 i.i.d. samples from the true posterior.
For privacy accounting, we used Theorem~\ref{gauss_composition_theorem} and
Theorem~\ref{dp_hmc_privacy_theorem} for DP-penalty and DP-HMC, respectively.
For DP-SGLD and DP-SGNHT, we used the Fourier accountant\footnote{
  We used the original implementation from \url{https://github.com/DPBayes/PLD-Accountant}.
} of \citet{KJH20} that computes tight privacy bounds for the subsampled Gaussian
mechanism. We used \(\delta = \frac{0.1}{n}\) for all runs, and varied \(\epsilon\).
We used a constant step size for DP-SGLD and DP-SGNHT as computing
privacy bounds for decreasing step size is infeasible with the Fourier accountant.

Log-likelihood ratio clip bounds for DP-penalty and
DP-HMC were tuned to have less than 20\% of the log-likelihood ratios clipped,
as the clipping experiment in Section~\ref{clipping_experiment_section} shows
that it leads to minimal effect on the posterior. We used the same guideline
for gradient clipping in DP-SGLD and DP-SGNHT, but did not experimentally verify the effect of
clipping for them. Gradient clipping for DP-HMC does not affect
asymptotic convergence, so it was tuned to minimise the effect of gradient
clipping and noise on the acceptance rate.

\paragraph{Results}
The top row of Figure~\ref{comparison_figure} shows the result of running each algorithm
on the banana model. DP-HMC and DP-penalty have roughly equal performance on
both MMD and mean error, while DP-SGLD and DP-SGNHT have significantly worse performance,
especially with the higher values of \(\epsilon\).
The bottom row shows the results with the Gaussian model.
The best performer was DP-SGLD, DP-HMC and DP-SGNHT were mostly equal, and
DP-penalty performed the worst.

Figure~\ref{dist_comparison_figure}
compares the posteriors from each algorithm with \(\epsilon = 15\) to the
true posterior on the banana model. The comparison shows the reason for the poor performance
of DP-SGLD and DP-SGNHT: they have trouble exploring the long tail of the
posterior. The median sample of DP-penalty highlights one of the difficulties in sampling
the banana model: the sample seems to cover the posterior well from the 2D plot,
but the marginal plots reveal that it overrepresents the tail, which is likely
a result of one of the chains getting stuck in the tail. DP-HMC is more consistent
in this regard.

\begin{figure}
	\centering
  \includegraphics[width=\linewidth]{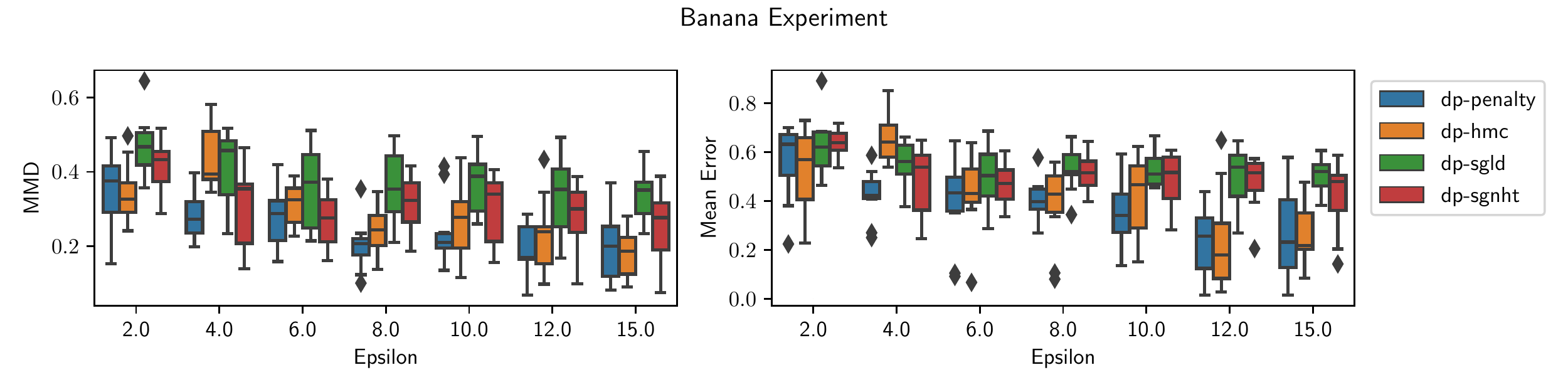}
  \includegraphics[width=\linewidth]{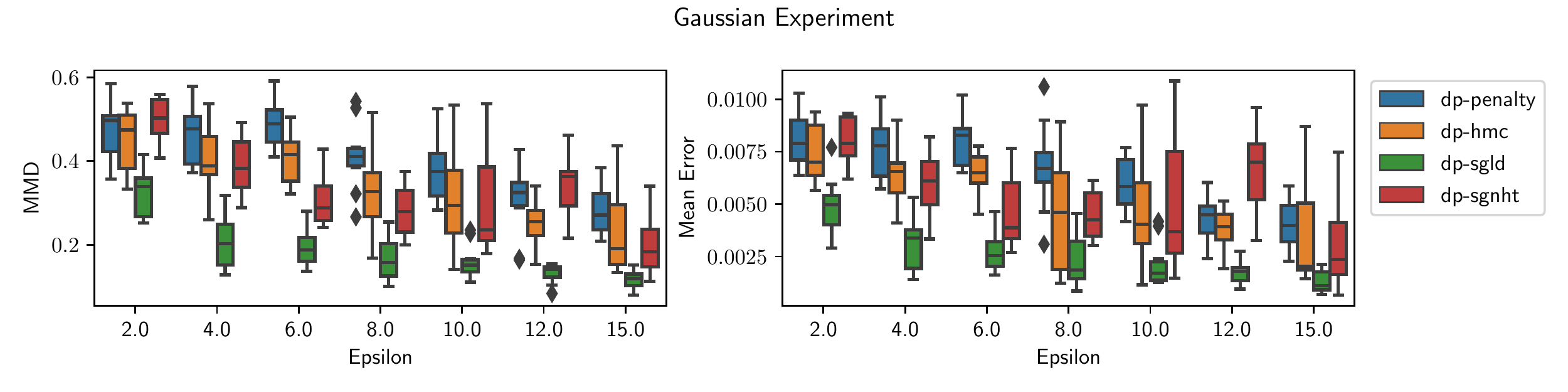}
  \caption{
    MMD and mean error for the banana and Gaussian models.
  }
  \label{comparison_figure}
\end{figure}

\begin{figure}
	\centering
  \includegraphics[width=\linewidth]{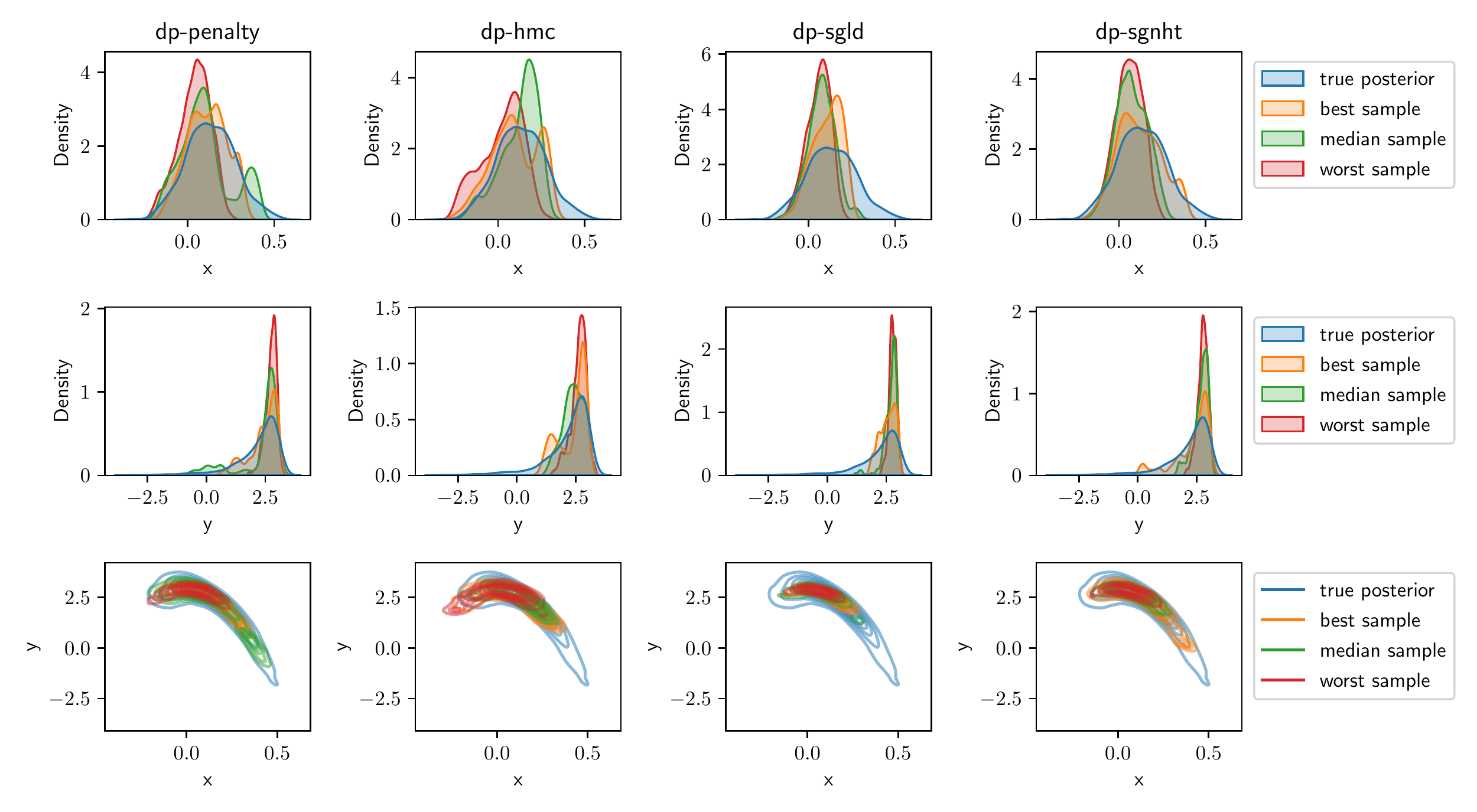}
  \caption{
    Visual comparison of the posteriors from each algorithm and the
    true posterior with \(\epsilon = 15\). The top and middle rows
    show KDE plots of the marginal distributions for the best, median and worst samples
    by MMD compared to the reference posterior sample. The bottom row shows a 2D KDE
    of each sample compared to the reference posterior sample.
  }
  \label{dist_comparison_figure}
\end{figure}

\subsection{Clipping Experiment}\label{clipping_experiment_section}

\paragraph{Implementation}
To asses the effect of clipping log-likelihood ratios, we ran random walk
Metropolis-Hastings (RWMH) and HMC on both the banana and Gaussian models
while clipping log-likelihood ratios. We did not add noise at any point,
and used a large gradient clip bound for HMC\footnote{
  We used a large gradient clip bound, as the leapfrog
  proposal sometimes diverges in the tails of the banana distribution, and
  doing some gradient clipping helps to mitigate the divergence.
}, to isolate the effect of log-likelihood ratio clipping.
We chose the number of iterations and parameters for the algorithms by
ensuring that they converge in the sense that \(\hat{R} < 1.05\)~\cite{BDA},
and ran the algorithms with varying clip bounds.
Otherwise, we used the same setup as with the main experiments in
Section~\ref{experiments_section}: we ran 4 chains for each triple of
clip bound, algorithm and model, and computed the MMD of the combined sample
from all chains, with the first half of each chain discarded.
Each run was repeated 10 times, with the same starting points as the main experiment.

\paragraph{Results}
Figure~\ref{clip_figure} shows the results of the clipping experiment.
With a large enough clip bound, there is very little effect on the MMD, as seen
on the left side panels. The right side panels show MMD as a function of the
fraction of log-likelihood ratios that were clipped, which shows that clipping
has very little effect when less than 20\% of the log-likelihood ratios are clipped,
which we used as our guideline for tuning the clip bounds for our main experiments.
Not all of runs converged, especially with the smaller clip bounds, as we only
set the parameters for the largest clip bound.

Based on the 20\% guideline, for the experiments of Section~\ref{experiments_section},
we set the clip bounds 0.1 and 6 for DP-HMC on the
banana and Gaussian, respectively, and 0.15 and 10 for DP-penalty on the banana
and Gaussian, respectively. Based on Figure~\ref{clip_figure}, there should be
minimal effect from clipping at those bounds.

There is an interesting contrast in the results for the banana and Gaussian
models in Figure~\ref{clip_figure}. On the Gaussian model, there is a gap
in the fractions of clipped log-likelihood ratios between 0.2 and 0.6, while
it is not present with the banana. This could be a result of fact that the Gaussian
experiment does not have any clip bounds between 1 and 5, which is a fairly
large jump.

\begin{figure}
	\centering
  \includegraphics[width=\linewidth]{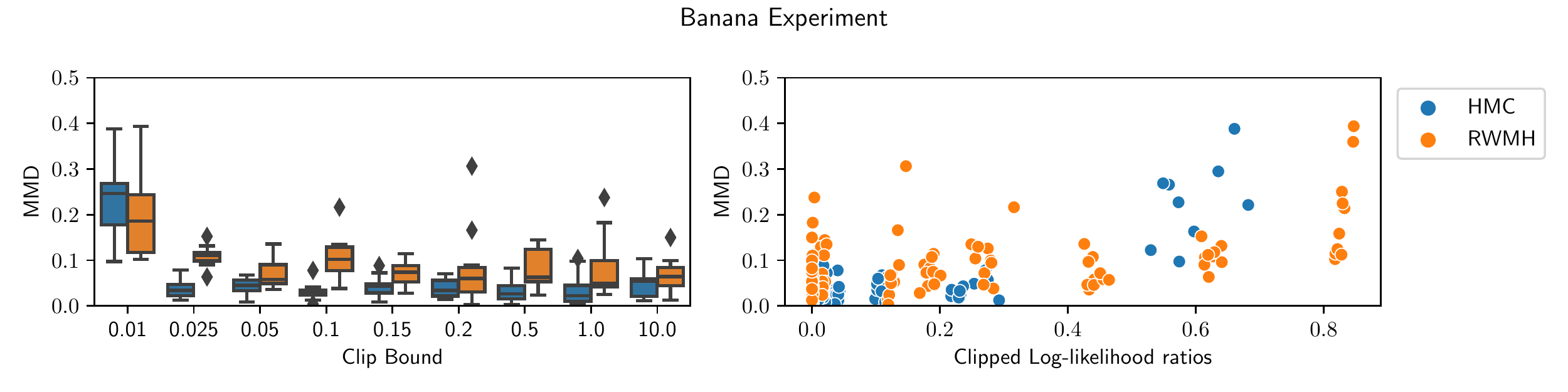}
  \includegraphics[width=\linewidth]{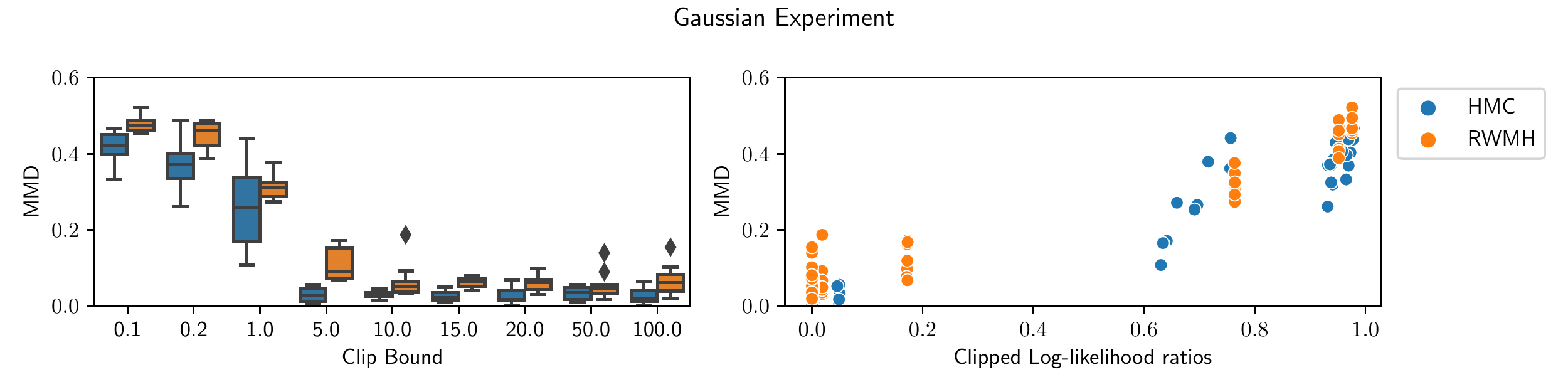}
  \caption{Results for the clipping experiment. The left side shows MMD as a
    function of clip bound, while the right side shows MMD as a function of
    the fraction of log-likelihood ratios that were clipped. Each point represents
    one of the 10 runs for each triple of clip bound, algorithm and model.
  }
  \label{clip_figure}
\end{figure}

\section{Discussion}\label{discussion_section}
\paragraph{Limitations}
Our experiment in Section~\ref{experiments_section} show that the MH acceptance
test is useful for efficient exploration of the tails of the banana posterior.
However, we had to use very large values of \(\epsilon\) to make any progress
towards sampling from the entire posterior, so it is clear that DP-MCMC methods
cannot achieve comparable performance to non-DP methods, unless very loose
privacy bounds are used, or a very large dataset \((n \gg 10^{5})\) is used.
We also noticed that the acceptance rates for DP-penalty and DP-HMC drop rapidly
with increasing dimension on the Gaussian model, which is why we used a fairly
low number of dimensions (\(d = 10\)). This is a major limitation that warrants
further investigation.

Another major limitation of our work is our reliance on the Gaussian mechanism,
which is likely vulnerable to floating point inaccuracies in computer implementations
that destroy the theoretical privacy guarantees~\cite{Mir12}. The discrete Gaussian
mechanism~\cite{C0S20} can be used in place of the Gaussian mechanism for many
applications, but the penalty algorithm requires adding Gaussian noise, so
the discrete Gaussian cannot be used as a drop-in replacement.

\paragraph{Future research}
There are many potential improvements to DP-HMC. Subsampling the gradients, as is done in DP-SGLD and
DP-SGNHT, would provide a significant privacy budget saving. However, naive gradient
subsampling is likely to lower acceptance rates significantly, especially in high dimensions~\cite{Bet15}.
The SGHMC~\cite{CFG14} and SGNHT~\cite{DFB14} algorithms correct for gradient
subsampling by adding friction to the Hamiltonian dynamics, but they forego the MH
acceptance test. Conducting the MH test with the added friction is not trivial,
but it has been done for SGHMC in the AMAGOLD algorithm~\cite{ZhangCS20}.

\citet{HeikkilaJDH19} used subsampling in the acceptance test of their
DP MCMC algorithm by assuming that the error from subsampling is close to
Gaussian by the central limit theorem. The same justification could be applied to
the penalty algorithm, but in our preliminary experiments it substantially lowered
the acceptance rate and did not improve the overall results.

Other potential improvements for DP-HMC are tuning the parameters,
especially \(\eta\) and \(L\), automatically. NUTS~\cite{HoG14} is the most
famous HMC variant that tunes \(\eta\) and \(L\) automatically, but it has
a very complicated sampling process. The recent ChEES-HMC algorithm~\cite{HoffmanRS21}
has a much simpler automatic tuning process, making it more suitable for integration
into DP-HMC.

\section{Conclusion}\label{conclusion_section}
We developed DP-HMC, a DP variant of HMC, and proved that it has the correct
invariant distribution and is ergodic in Section~\ref{dp_hmc_section}.
In Section~\ref{experiments_section}, we compared DP-HMC with existing DP-MCMC
algorithms, and showed that DP-HMC is consistently better or equal to DP-penalty,
while DP-SGLD and DP-SGNHT did not perform consistently.

\section*{Acknowledgements}
We would like to thank Eero Saksman for his thoughts on DP-HMC which
inspired our measure-theoretic convergence proof.
This work has been supported by the Academy of Finland (Finnish Center for
Artificial Intelligence FCAI and grant 325573) as well as by the Strategic
Research Council at the Academy of Finland (grant 336032).

\bibliography{references}

\begin{thebibliography}{30}
\providecommand{\natexlab}[1]{#1}
\providecommand{\url}[1]{\texttt{#1}}
\expandafter\ifx\csname urlstyle\endcsname\relax
  \providecommand{\doi}[1]{doi: #1}\else
  \providecommand{\doi}{doi: \begingroup \urlstyle{rm}\Url}\fi

\bibitem[Betancourt(2015)]{Bet15}
Michael Betancourt.
\newblock The fundamental incompatibility of scalable {H}amiltonian {M}onte
  {C}arlo and naive data subsampling.
\newblock In \emph{Proceedings of the 32nd International Conference on Machine
  Learning}, volume~37 of \emph{{JMLR} Workshop and Conference Proceedings},
  pages 533--540. 2015.

\bibitem[Canonne et~al.(2020)Canonne, Kamath, and Steinke]{C0S20}
Cl{\'{e}}ment~L. Canonne, Gautam Kamath, and Thomas Steinke.
\newblock The discrete {G}aussian for differential privacy.
\newblock In \emph{Advances in Neural Information Processing Systems 33: Annual
  Conference on Neural Information Processing Systems}, 2020.

\bibitem[Ceperley and Dewing(1999)]{CeD99}
DM~Ceperley and Mark Dewing.
\newblock The penalty method for random walks with uncertain energies.
\newblock \emph{The Journal of chemical physics}, 110\penalty0 (20):\penalty0
  9812--9820, 1999.

\bibitem[Chen et~al.(2014)Chen, Fox, and Guestrin]{CFG14}
Tianqi Chen, Emily~B. Fox, and Carlos Guestrin.
\newblock Stochastic gradient {H}amiltonian {M}onte {C}arlo.
\newblock In \emph{Proceedings of the 31th International Conference on Machine
  Learning}, volume~32 of \emph{{JMLR} Workshop and Conference Proceedings},
  pages 1683--1691. 2014.

\bibitem[{\c{C}}{\i}nlar(2011)]{Cin11}
Erhan {\c{C}}{\i}nlar.
\newblock \emph{Probability and Stochastics}.
\newblock Graduate Texts in Mathematics, 261. Springer New York, New York, NY,
  1st edition, 2011.

\bibitem[Ding et~al.(2014)Ding, Fang, Babbush, Chen, Skeel, and Neven]{DFB14}
Nan Ding, Youhan Fang, Ryan Babbush, Changyou Chen, Robert~D. Skeel, and
  Hartmut Neven.
\newblock Bayesian sampling using stochastic gradient thermostats.
\newblock In \emph{Advances in Neural Information Processing Systems 27: Annual
  Conference on Neural Information Processing Systems}, pages 3203--3211, 2014.

\bibitem[Duane et~al.(1987)Duane, Kennedy, Pendleton, and Roweth]{DKP87}
Simon Duane, Anthony~D Kennedy, Brian~J Pendleton, and Duncan Roweth.
\newblock Hybrid {M}onte {C}arlo.
\newblock \emph{Physics letters B}, 195\penalty0 (2):\penalty0 216--222, 1987.

\bibitem[Durmus et~al.(2020)Durmus, Moulines, and Saksman]{DMS20}
Alain Durmus, Eric Moulines, and Eero Saksman.
\newblock Irreducibility and geometric ergodicity of {H}amiltonian {M}onte
  {C}arlo.
\newblock \emph{Annals of Statistics}, 48\penalty0 (6):\penalty0 3545--3564,
  2020.

\bibitem[Dwork and Roth(2014)]{DwR14}
Cynthia Dwork and Aaron Roth.
\newblock The algorithmic foundations of differential privacy.
\newblock \emph{Foundations and Trends in Theoretical Computer Science},
  9\penalty0 (3-4):\penalty0 211--407, 2014.

\bibitem[Dwork et~al.(2006{\natexlab{a}})Dwork, Kenthapadi, McSherry, Mironov,
  and Naor]{DKM06}
Cynthia Dwork, Krishnaram Kenthapadi, Frank McSherry, Ilya Mironov, and Moni
  Naor.
\newblock Our data, ourselves: Privacy via distributed noise generation.
\newblock In \emph{Advances in Cryptology - {EUROCRYPT} 2006, 25th Annual
  International Conference on the Theory and Applications of Cryptographic
  Techniques}, volume 4004 of \emph{Lecture Notes in Computer Science}, pages
  486--503. 2006{\natexlab{a}}.

\bibitem[Dwork et~al.(2006{\natexlab{b}})Dwork, McSherry, Nissim, and
  Smith]{DMN06}
Cynthia Dwork, Frank McSherry, Kobbi Nissim, and Adam~D. Smith.
\newblock Calibrating noise to sensitivity in private data analysis.
\newblock In \emph{Theory of Cryptography, Third Theory of Cryptography
  Conference, {TCC}}, volume 3876 of \emph{Lecture Notes in Computer Science},
  pages 265--284. 2006{\natexlab{b}}.

\bibitem[Gelman et~al.(2014)Gelman, Carlin, Stern, Dunson, Vehtari, and
  Rubin]{BDA}
Andrew Gelman, John~B Carlin, Hal~S Stern, David~B Dunson, Aki Vehtari, and
  Donald~B Rubin.
\newblock \emph{{B}ayesian data analysis}.
\newblock Chapman \& Hall/CRC texts in statistical science series. CRC Press,
  Boca Raton, third edition, 2014.

\bibitem[Gretton et~al.(2012)Gretton, Borgwardt, Rasch, Sch{\"{o}}lkopf, and
  Smola]{GrettonBRSS12}
Arthur Gretton, Karsten~M. Borgwardt, Malte~J. Rasch, Bernhard Sch{\"{o}}lkopf,
  and Alexander~J. Smola.
\newblock A kernel two-sample test.
\newblock \emph{J. Mach. Learn. Res.}, 13:\penalty0 723--773, 2012.

\bibitem[Hastings(1970)]{Has70}
W.K. Hastings.
\newblock {M}onte {C}arlo sampling methods using {M}arkov chains and their
  applications.
\newblock \emph{Biometrika}, 57\penalty0 (1):\penalty0 97--109, 1970.

\bibitem[Heikkil{\"{a}} et~al.(2019)Heikkil{\"{a}}, J{\"{a}}lk{\"{o}}, Dikmen,
  and Honkela]{HeikkilaJDH19}
Mikko~A. Heikkil{\"{a}}, Joonas J{\"{a}}lk{\"{o}}, Onur Dikmen, and Antti
  Honkela.
\newblock Differentially private {M}arkov chain {M}onte {C}arlo.
\newblock In \emph{Advances in Neural Information Processing Systems 32: Annual
  Conference on Neural Information Processing Systems}, pages 4115--4125, 2019.

\bibitem[Hoffman et~al.(2021)Hoffman, Radul, and Sountsov]{HoffmanRS21}
Matthew Hoffman, Alexey Radul, and Pavel Sountsov.
\newblock An adaptive-{MCMC} scheme for setting trajectory lengths in
  {H}amiltonian {M}onte {C}arlo.
\newblock In \emph{The 24th International Conference on Artificial Intelligence
  and Statistics, {AISTATS}}, volume 130 of \emph{Proceedings of Machine
  Learning Research}, pages 3907--3915. 2021.

\bibitem[Hoffman and Gelman(2014)]{HoG14}
Matthew~D Hoffman and Andrew Gelman.
\newblock The {No-U-Turn} sampler: adaptively setting path lengths in
  {H}amiltonian {M}onte {C}arlo.
\newblock \emph{J. Mach. Learn. Res.}, 15\penalty0 (1):\penalty0 1593--1623,
  2014.

\bibitem[Koskela et~al.(2020)Koskela, J{\"{a}}lk{\"{o}}, and Honkela]{KJH20}
Antti Koskela, Joonas J{\"{a}}lk{\"{o}}, and Antti Honkela.
\newblock Computing tight differential privacy guarantees using {FFT}.
\newblock In \emph{The 23rd International Conference on Artificial Intelligence
  and Statistics, {AISTATS}}, volume 108 of \emph{Proceedings of Machine
  Learning Research}, pages 2560--2569. 2020.

\bibitem[Li et~al.(2019)Li, Chen, Liu, and Carin]{LCL19}
Bai Li, Changyou Chen, Hao Liu, and Lawrence Carin.
\newblock On connecting stochastic gradient {MCMC} and differential privacy.
\newblock In \emph{Proceedings of Machine Learning Research}, volume~89 of
  \emph{Proceedings of Machine Learning Research}, pages 557--566. 16--18 Apr
  2019.

\bibitem[Metropolis et~al.(1953)Metropolis, Rosenbluth, Rosenbluth, Teller, and
  Teller]{MRR53}
Nicholas Metropolis, Arianna~W Rosenbluth, Marshall~N Rosenbluth, Augusta~H
  Teller, and Edward Teller.
\newblock Equation of state calculations by fast computing machines.
\newblock \emph{The journal of chemical physics}, 21\penalty0 (6):\penalty0
  1087--1092, 1953.

\bibitem[Mironov(2012)]{Mir12}
Ilya Mironov.
\newblock On significance of the least significant bits for differential
  privacy.
\newblock In \emph{the {ACM} Conference on Computer and Communications
  Security, CCS'12}, pages 650--661. 2012.

\bibitem[Neal(2011)]{neal2012mcmc}
Radford~M. Neal.
\newblock {MCMC} using {H}amiltonian dynamics.
\newblock In \emph{Handbook of {M}arkov Chain {M}onte {C}arlo}. Chapman {\&}
  Hall {/} CRC Press, 2011.

\bibitem[Owen(2017)]{Owen17tech}
Art~B Owen.
\newblock A randomized {H}alton algorithm in {R}.
\newblock Technical report, Stanford University, 2017.
\newblock arXiv:1706.02808.

\bibitem[Robert and Casella(2004)]{Robert04}
Christian~P. Robert and George Casella.
\newblock \emph{{M}onte {C}arlo statistical methods}.
\newblock Springer texts in statistics. Springer, New York, 2nd ed. edition,
  2004.

\bibitem[Sommer et~al.(2019)Sommer, Meiser, and Mohammadi]{Sommer2019}
David~M. Sommer, Sebastian Meiser, and Esfandiar Mohammadi.
\newblock Privacy loss classes: The central limit theorem in differential
  privacy.
\newblock \emph{PoPETs}, 2019\penalty0 (2):\penalty0 245--269, 2019.

\bibitem[Tierney(1998)]{Tie98}
Luke Tierney.
\newblock A note on {M}etropolis-{H}astings kernels for general state spaces.
\newblock \emph{Annals of applied probability}, pages 1--9, 1998.

\bibitem[Tran et~al.(2014)Tran, Pitt, and Kohn]{TPK14}
Minh{-}Ngoc Tran, Michael~K. Pitt, and Robert Kohn.
\newblock Adaptive {M}etropolis-{H}astings sampling using reversible dependent
  mixture proposals.
\newblock \emph{Statistics and Computing}, 26\penalty0 (1-2):\penalty0
  361--381, 2014.

\bibitem[Wang et~al.(2015)Wang, Fienberg, and Smola]{WFS15}
Yu{-}Xiang Wang, Stephen~E. Fienberg, and Alexander~J. Smola.
\newblock Privacy for free: Posterior sampling and stochastic gradient {M}onte
  {C}arlo.
\newblock In \emph{Proceedings of the 32nd International Conference on Machine
  Learning, {ICML}}, volume~37 of \emph{{JMLR} Workshop and Conference
  Proceedings}, pages 2493--2502. 2015.

\bibitem[Yildirim and Ermis(2019)]{YildirimE19}
Sinan Yildirim and Beyza Ermis.
\newblock Exact {MCMC} with differentially private moves - revisiting the
  penalty algorithm in a data privacy framework.
\newblock \emph{Statistics and Computing}, 29\penalty0 (5):\penalty0 947--963,
  2019.

\bibitem[Zhang et~al.(2020)Zhang, Cooper, and Sa]{ZhangCS20}
Ruqi Zhang, A.~Feder Cooper, and Christopher~De Sa.
\newblock {AMAGOLD:} amortized metropolis adjustment for efficient stochastic
  gradient {MCMC}.
\newblock In \emph{The 23rd International Conference on Artificial Intelligence
  and Statistics, {AISTATS}}, volume 108 of \emph{Proceedings of Machine
  Learning Research}, pages 2142--2152. 2020.

\end{thebibliography}

\appendix
\renewcommand{\appendixpagename}{
  Suppelementary Material for Differentially Private Hamiltonian Monte Carlo
}
\appendixpage

\section{Measure Theory}\label{measure_theory_section}

In this section, we prove the measure-theoretic results stated in the main text
but not proved there. We start by recalling the main definitions
of Section~\ref{hmc_convergence_section}:
\markovkerneldefinition*
\markovkernelreversible*

\begin{lemma}\label{markov_kernel_integral_lemma}
	Let \(q_{1}\) and \(q_{2}\) be Markov kernels on \((E, \cale)\), let
  \(\mu\) be a \(\sigma\)-finite measure, and let \(f\colon E\times E \to \R_{+}\)
  be a measurable function. If
  \[
    \int_{A}\mu(\dx a)\int_{B}q_{1}(a, \dx b)
    = \int_{B}\mu(\dx b)\int_{A}q_{2}(b, \dx a),
  \]
  for all \(A, B\in \cale\),
  then
  \[
    \int_{A}\mu(\dx a)\int_{B}q_{1}(a, \dx b)f(a, b)
    = \int_{B}\mu(\dx b)\int_{A}q_{2}(b, \dx a)f(a, b)
  \]
  for all \(A, B\in \cale\).
\end{lemma}
\begin{proof}
  The condition
  \[
    \int_{A}\mu(\dx a)\int_{B}q_{1}(a, \dx b)
    = \int_{B}\mu(\dx b)\int_{A}q_{2}(b, \dx a)
  \]
  means that the measures (as in Lemma~\ref{markov_kernel_product_unique_lemma})
  \[
    \nu_{1}(A\times B) = \int_{A}\mu(\dx a)\int_{B}q(a, \dx b),
  \]
  \[
    \nu_{2}(A\times B) = \int_{B}\mu(\dx b)\int_{A}q(b, \dx a)
  \]
  are equal. Then
  \begin{align*}
    \int_{A}\mu(\dx a)\int_{B}q_{1}(a, \dx b)f(a, b)
    &= \int_{A\times B}\nu_{1}(\dx a, \dx b)f(a, b)
    \\&= \int_{A\times B}\nu_{2}(\dx a, \dx b)f(a, b)
    \\&= \int_{B}\mu(\dx b)\int_{A}q_{2}(b, \dx a)f(a, b)
  \end{align*}
  for all \(A, B\in \cale\).
\end{proof}
\begin{corollary}\label{markov_kernel_reversible_integral_corollary}
  Let \(q\) be a Markov kernel on \((E, \cale)\) reversible with respect to
  a \(\sigma\)-finite measure \(\mu\). Then
  \[
    \int_{A}\mu(\dx a)\int_{B}q(a, \dx b)f(a, b)
    = \int_{B}\mu(\dx b)\int_{A}q(b, \dx a)f(a, b)
  \]
  for all \(A, B\in \cale\).
\end{corollary}
\begin{proof}
	The claim follows by setting \(q_{1} = q_{2} = q\) in
  Lemma~\ref{markov_kernel_integral_lemma}, as the condition of
  Lemma~\ref{markov_kernel_integral_lemma} is then the reversibility of
  \(q\) with respect to \(\mu\).
\end{proof}

\markovkernelreversiblecomposition*
\begin{proof}
  We prove the claim by induction on \(k\). For \(k = 1\), the claim is the
  definition of reversibility of \(q_{1}\) with respect to \(\mu\).
  If the claim holds for \(k - 1\), for any \(A, C\in \cale\),
  \begin{align}
    \int_{A}\mu(\dx a)\int_{C}(q_{k}\circ \dotsb \circ q_{1})(a, \dx c)
    &= \int_{A}\mu(\dx a)\int_{E}(q_{k-1}\circ \dotsb \circ q_{1})(a, \dx b)\int_{C}q_{k}(b, \dx c)
      \label{mkrc_line1}
    \\&= \int_{E}\mu(\dx b)\int_{A}(q_{1}\circ \dotsb \circ q_{k-1})(b, \dx a)\int_{C}q_{k}(b, \dx c)
      \label{mkrc_line2}
    \\&= \int_{E}\mu(\dx b)\int_{C}q_{k}(b, \dx c)\int_{A}(q_{1}\circ \dotsb \circ q_{k-1})(b, \dx a)
      \label{mkrc_line3}
    \\&= \int_{C}\mu(\dx c)\int_{E}q_{k}(c, \dx b)\int_{A}(q_{1}\circ \dotsb \circ q_{k-1})(b, \dx a)
      \label{mkrc_line4}
    \\&= \int_{C}\mu(\dx c)\int_{A}(q_{1}\circ \dotsb \circ q_{k})(c, \dx a).
      \label{mkrc_line5}
  \end{align}
  Lines~(\ref{mkrc_line1}) and (\ref{mkrc_line5}) follow from
  Lemma~\ref{markov_kernel_composition_lemma},
  line~(\ref{mkrc_line2}) from the induction hypothesis and
  Lemma~\ref{markov_kernel_integral_lemma},
  and line (\ref{mkrc_line4}) from
  Corollary~\ref{markov_kernel_reversible_integral_corollary}.
\end{proof}

\mhreversibleproposal*
\begin{proof}
  For acceptance probability \(\alpha\), the detailed balance condition
  \[
    \int_{A}\pi(\dx \theta)\int_{B}q(\theta, \dx \theta')\alpha(\theta, \theta')
    = \int_{B}\pi(\dx \theta')\int_{A}q(\theta', \dx \theta)\alpha(\theta', \theta)
  \]
  for all measurable \(A, B\subset \R^{d}\) implies the invariance of \(\pi\)~\cite{Tie98}.\footnote{
    \citet{Tie98} states the detailed balance condition as an equality of
    measures, which is equivalent to the stated equality of integrals by
    Lemma~\ref{markov_kernel_product_unique_lemma}.
  }
  If \(\pi\) is continuous and \(q\) is reversible with respect to the Lebesgue
  measure \(m\), for measurable \(A, B\subset \R^{d}\):
  \begin{align*}
    \int_{A}\pi(\dx \theta)\int_{B}q(\theta, \dx \theta')\alpha(\theta, \theta')
    &= \int_{A}m(\dx \theta)\int_{B}\pi(\theta)q(\theta, \dx \theta')\alpha(\theta, \theta')
    \\&= \int_{A}m(\dx \theta)\int_{B}q(\theta, \dx \theta')\min\{\pi(\theta), \pi(\theta')\}
    \\&= \int_{B}m(\dx \theta')\int_{A}q(\theta', \dx \theta)\min\{\pi(\theta'), \pi(\theta)\}
    \\&= \int_{B}\pi(\dx \theta')\int_{A}q(\theta', \dx \theta)\alpha(\theta', \theta),
  \end{align*}
  which implies the invariance of \(\pi\).
\end{proof}

\diracreversible*
\begin{proof}
  As \(f = f^{-1}\) and preserves Lebesgue measure, for all measurable \(A, B\subset \R^{d}\):
  \begin{align*}
    \int_{A}m(\dx a)\int_{B}\delta_{f(a)}(\dx b)
    &= \int_{A}m(\dx a)1_{B}(f(a))
    \\&= m(A\cap f^{-1}(B))
    \\&= m(f^{-1}(A\cap f^{-1}(B)))
    \\&= m(f^{-1}(A)\cap B)
    \\&= \int_{B}m(\dx b)\int_{A}\delta_{f(b)}(\dx a). \qedhere
  \end{align*}
\end{proof}

For the convergence proof of DP-HMC, specifically
Lemma~\ref{dp_hmc_leapfrog_step_reversibility_lemma}, we must deal with Markov kernels defined
on \(\R^{2d}\) that have the auxiliary variable \(p\) in addition to the
parameter \(\theta\). The preceding theory cannot deal with both variables
separately, so we must develop theory that can, which culminates in
Lemma~\ref{markov_kernel_reversible_double_lemma}.

\begin{definition}\label{p_system_definition}
	Let \(E\) be a set. A collection \(\mathcal{C}\subset \mathcal{P}(E)\) is called a
  p-system if \(A\cap B\in \mathcal{C}\) for all \(A, B\in \mathcal{C}\).
\end{definition}

\begin{lemma}\label{measure_generator_equality_lemma}
  Let \(E\) be a set and let \(\mathcal{C}\subset \mathcal{P}(E)\) be a
  p-system. Let \(\cale\) be the \(\sigma\)-algebra generated by \(\mathcal{C}\).
  Let \(\mu\) and \(\nu\) be finite measures on \((E, \cale)\). If
  \(\mu(A) = \nu(A)\) for all \(A\in \mathcal{C}\), \(\mu = \nu\).
\end{lemma}
\begin{proof}
  See \citet[Proposition I.3.7]{Cin11}.
\end{proof}

\begin{lemma}\label{product_measure_equality_lemma}
  Let \((E, \cale)\) be a measurable space and let \(\mu\) and \(\nu\)
  be measures on \((E, \cale)^{d}\) with a countable partition \(P\) of \(E\)
  such that \(\mu(\bigtimes_{j=1}^{d}B_{j}) < \infty\) and
  \(\nu(\bigtimes_{j=1}^{d}B_{j}) < \infty\) for all \(B_{1},\dotsc, B_{d}\in P\).
  If
  \[
    \mu\left(\bigtimes_{j}^{d}A_{j}\right) = \nu\left(\bigtimes_{j}^{d}A_{j}\right)
  \]
  for all \(A_{1},\dotsc,A_{d}\in \cale\), \(\mu = \nu\).
\end{lemma}
\begin{proof}
  Let \(P_{d} = \{\bigtimes_{j=1}^{d}B_{j}\mid B_{1},\dotsc, B_{d}\in P\}\).
  Denote the \emph{restriction} of \(\mu\) into \(C\) by \(\mu|C\), which is
  the measure \((\mu|C)(A) = \mu(A\cap C)\)~\cite{Cin11}. The measures
  \(\mu|C\) and \(\nu|C\) for \(C \in P_{d}\) are finite as
  \((\mu|C)(A) \leq \mu(C) < \infty\) for any \(A\in \cale^{d}\) and the same
  holds for \(\nu\).

  Recall that \(\cale^{d}\) is generated by the p-system of sets of the
  form \(\bigtimes_{j=1}^{d}A_{j}\) for \(A_{1},\dotsc,A_{d}\in \cale\).
  For any \(C\in P_{d}\) and \(A_{1},\dotsc,A_{d}\in \cale\), we have
  \begin{align*}
    (\mu|C)\left(\bigtimes_{j=1}^{d}A_{j}\right)
    &= \mu\left(\left(\bigtimes_{j=1}^{d}A_{j}\right)
      \cap \left(\bigtimes_{j=1}^{d}B_{j}\right)\right)
    \\&= \mu\left(\bigtimes_{j=1}^{d}(A_{j}\cap B_{j})\right)
    \\&= \nu\left(\bigtimes_{j=1}^{d}(A_{j}\cap B_{j})\right)
    \\&= (\nu|B)\left(\bigtimes_{j=1}^{d}A_{j}\right),
  \end{align*}
  so \((\mu|C) = (\nu|C)\) for any \(C\in P_{d}\) by
  Lemma~\ref{measure_generator_equality_lemma}.

  As \(P\) is countable, the sets
  in \(P_{d}\) can be enumerated as \(C_{i}\) for \(i\in \N\). Now
  \[
    \mu(A) = \mu(E^{d}\cap A)
    = \mu\left(\bigcup_{i=1}^{\infty}(C_{i}\cap A)\right)
    = \sum_{i=1}^{\infty}(\mu|C_{i})(A)
    = \sum_{i=1}^{\infty}(\nu|C_{i})(A)
    = \nu(A)
  \]
  for any \(A\in \cale^{d}\), so \(\mu = \nu\).
\end{proof}

\begin{lemma}\label{markov_kernel_reversible_double_lemma}
	Let \((E, \cale)\) be a measurable space, let \(q\) be a Markov kernel on
  \((E, \cale)^{2}\) and let \(\mu\) be a \(\sigma\)-finite measure on
  \((E, \cale)\). Then \(q\) is reversible with respect to \(\mu^{2}\) if
  and only if
  \[
    \int_{A}\mu(\dx a)\int_{B}\mu(\dx b)\int_{C\times D}q((a, b), \dx(c, d))
    = \int_{C}\mu(\dx c)\int_{D}\mu(\dx d)\int_{A\times B}q((c, d), \dx(a, b))
  \]
  for all \(A, B, C, D\in E\).
\end{lemma}
\begin{proof}
  Let \(V, W\in \cale^{2}\) and
  \[
    \nu_{1}(V\times W) = \int_{V}\mu^{2}(\dx v)\int_{W}q(v, \dx w),
  \]
  \[
    \nu_{2}(V\times W) = \int_{W}\mu^{2}(\dx w)\int_{V}q(w, \dx v).
  \]
  Now reversibility of \(q\) with respect to \(\mu^{2}\) is equivalent
  to \(\nu_{1} = \nu_{2}\).

  If \(\nu_{1} = \nu_{2}\), for all \(A, B, C, D\in \cale\),
  \begin{align*}
    \int_{A}\mu(\dx a)\int_{B}\mu(\dx b)\int_{C\times D}q((a, b), \dx(c, d))
    &= \nu_{1}(A\times B\times C\times D)
    \\&= \nu_{2}(A\times B\times C\times D)
    \\&= \int_{C}\mu(\dx c)\int_{D}\mu(\dx d)\int_{A\times B}q((c, d), \dx(a, b)).
  \end{align*}

  If
  \[
    \int_{A}\mu(\dx a)\int_{B}\mu(\dx b)\int_{C\times D}q((a, b), \dx(c, d))
    = \int_{C}\mu(\dx c)\int_{D}\mu(\dx d)\int_{A\times B}q((c, d), \dx(a, b)),
  \]
  then
  \[
    \nu_{1}(A\times B\times C\times D) = \nu_{2}(A\times B\times C\times D)
  \]
  for all \(A, B, C, D\in \cale\).
  As \(\mu\) is \(\sigma\)-finite, there is a countable partition \(E_{i}\) of \(E\)
  such that \(\mu(E_{i})< \infty\) for all \(i\in \N\). Additionally,
  \[
    \nu_{1}(E_{i}\times E_{j}\times E_{k}\times E_{l})
    \leq \int_{E_{i}}\mu(\dx a)\int_{E_{j}}\mu(\dx b)
    < \infty
  \]
  and
  \[
    \nu_{2}(E_{i}\times E_{j}\times E_{k}\times E_{l})
    \leq \int_{E_{k}}\mu(\dx c)\int_{E_{l}}\mu(\dx d)
    < \infty.
  \]
  so \(\nu_{1} = \nu_{2}\) by Lemma~\ref{product_measure_equality_lemma}.
\end{proof}

\dphmcleapfrogstepreversibility*
\begin{proof}
  Starting with \(l_{\theta}^{-} = l_{\theta}\circ l_{-}\), note that \(l_{\theta}\circ l_{-}\)
  is an involution that preserves Lebesgue measure. The Markov kernel for
  \(l_{\theta}\circ l_{-}\) is \(\delta_{(l_{\theta}\circ l_{-}})(\theta, p)\),
  so the claim follows from Lemma~\ref{dirac_reversible_lemma}.

  Recall that both \(l_{p_{\nicefrac{\eta}{2}}}^{-}\) and \(l_{p_{\eta}}^{-}\) are of the form
  \[
    (l_{-}\circ l_{p_{s}})(\theta, p) = (\theta, -p + s(g(\theta) + \xi)),
  \]
  where \(\xi\sim \caln(0, \sigma_{g}^{2})\) and \(s > 0\).
  Definition~\ref{markov_kernel_reversible_definition} for
  \(l_{-}\circ l_{p_{s}}\) is
  \[
    \int_{V}m_{2d}(\dx v)\int_{W}(l_{-}\circ l_{p_{s}})(v, \dx w)
    = \int_{W}m_{2d}(\dx w)\int_{V}(l_{-}\circ l_{p_{s}})(w, \dx v)
  \]
  for all measurable \(V, W\in \R^{2d}\). Because of Lemma~\ref{markov_kernel_reversible_double_lemma},
  this can be stated as
  \begin{align*}
    &\int_{A}m_{d}(\dx \theta)\int_{B}m(\dx p)\int_{C\times D}
    (l_{-}\circ l_{p_{s}})(\theta, p, \dx(\theta', p'))
    \\= &\int_{C}m_{d}(\dx \theta')\int_{D}m(\dx p')\int_{A\times B}
    (l_{-}\circ l_{p_{s}})(\theta', p', \dx(\theta, p))
  \end{align*}
  for all measurable \(A, B, C, D\subset \R^{d}\). Denote \(f(\theta) = sg(\theta)\),
  \(\sigma^{2} = s^{2}\sigma_{g}^{2}\) and the density function of the
  \(d\)-dimensional Gaussian distribution by \(\caln(\cdot\mid \mu, \Sigma)\).
  Now, for any measurable \(C, D\subset \R^{d}\)
  \begin{align}
    \int_{C\times D}(l_{-}\circ l_{p_{s}})(\theta, p, \dx(\theta', p'))
    &= 1_{C}(\theta)\int_{D}m_{d}(\dx p')\caln_{d}(p'\mid -p + f(\theta), \sigma^{2}I)
      \label{l_p_equality}
    \\&= \int_{C}\delta_{\theta}(\dx \theta')
    \int_{D}m_{d}(\dx p')\caln_{d}(p'\mid -p + f(\theta), \sigma^{2}I),
  \end{align}
  which leads into
  \begin{align}
    &\int_{A}m_{d}(\dx \theta)\int_{B}m_{d}(\dx p)\int_{C\times D}
    (l_{-}\circ l_{p_{s}})(\theta, p, \dx(\theta', p'))
    \\&= \int_{A}m_{d}(\dx \theta)\int_{B}m_{d}(\dx p)
    \int_{C}\delta_{\theta}(\dx \theta')\int_{D}m_{d}(\dx p')\caln_{d}(p'\mid -p + f(\theta), \sigma^{2}I)
    \\&= \int_{A}m_{d}(\dx \theta)\int_{C}\delta_{\theta}(\dx \theta')\int_{B}m_{d}(\dx p)
    \int_{D}m_{d}(\dx p')\caln_{d}(p'\mid -p + f(\theta), \sigma^{2}I)
    \\&= \int_{C}m_{d}(\dx \theta')\int_{A}\delta_{\theta'}(\dx \theta)\int_{B}m_{d}(\dx p)
    \int_{D}m_{d}(\dx p')\caln_{d}(p'\mid -p + f(\theta), \sigma^{2}I)
    \label{dirac_reversibility_step}
    \\&= \int_{C}m_{d}(\dx \theta')\int_{A}\delta_{\theta'}(\dx \theta)\int_{B}m_{d}(\dx p)
    \int_{D}m_{d}(\dx p')\caln_{d}(p\mid -p' + f(\theta), \sigma^{2}I)
    \\&= \int_{C}m_{d}(\dx \theta')1_{A}(\theta')\int_{B}m_{d}(\dx p)
    \int_{D}m_{d}(\dx p')\caln_{d}(p\mid -p' + f(\theta'), \sigma^{2}I)
    \label{dirac_evaluation_step}
    \\&= \int_{C}m_{d}(\dx \theta')\int_{D}m_{d}(\dx p')1_{A}(\theta')
    \int_{B}m_{d}(\dx p)\caln_{d}(p\mid -p' + f(\theta'), \sigma^{2}I)
    \\&= \int_{C}m_{d}(\dx \theta')\int_{D}m(\dx p')\int_{A\times B}
    (l_{-}\circ l_{p_{s}})(\theta', p', \dx(\theta, p)),
    \label{l_p_equality_step}
  \end{align}
  where (\ref{dirac_reversibility_step}) uses Lemma~\ref{markov_kernel_reversible_integral_corollary},
  (\ref{dirac_evaluation_step}) uses the property of the Dirac measure that
  \(\int_{A}\delta_{b}(\dx a)f(a) = 1_{A}(b)f(b)\) for \(f\colon \R^{d}\to \R^{d}\)
  and (\ref{l_p_equality_step}) uses Equation~(\ref{l_p_equality}).
\end{proof}

\end{document}